\documentclass[12pt,a4paper]{article}
\usepackage{}

\usepackage{latexsym}
\usepackage{amsmath}
\usepackage{amssymb}
\usepackage{arydshln}

\usepackage{cite}
\usepackage{stmaryrd}
\usepackage{enumerate}

\usepackage{hyperref}

\usepackage{color}
\usepackage{lineno}
\usepackage{graphicx}
\usepackage{ae}
\usepackage{amsmath}
\usepackage{amssymb}
\usepackage{latexsym}
\usepackage{url}
\usepackage{epsfig}
\usepackage{mathrsfs}
\usepackage{amsfonts}
\usepackage{amsthm}

\usepackage{float}
\usepackage{subfig}
\usepackage{tikz}
\usetikzlibrary{positioning}

\newcommand{\Cay}{\mathrm{Cay}}

\newtheorem{theorem}{Theorem}[section]

\newtheorem{lemma}[theorem]{Lemma}

\newtheorem{cor}[theorem]{Corollary}

\newtheorem{example}{Example}

\theoremstyle{definition}

\numberwithin{equation}{section} 
\allowdisplaybreaks    

\def\qed{\hfill$\Box$\vspace{12pt}}

\long\def\delete#1{}

\usepackage{xcolor}
\usepackage[normalem]{ulem}

\usepackage{tikz}
\usetikzlibrary{decorations.markings}
\tikzstyle{vertex}=[circle, inner sep=1.2pt, minimum size=3pt]
\tikzstyle{filledvertex}=[circle, draw, fill, inner sep=1.2pt, minimum size=3pt]

\newcommand{\vertex}{\node[vertex]}

\tikzstyle{directed}=[postaction={decorate,
	decoration={markings,mark=at position 0.5 with {\arrow{stealth}}}
}]

\begin{document}
\title {Constructing graphs having Laplacian pair state transfer by an edge perturbation}

\author{Wei Wang$^{a,b}$,~Xiaogang Liu$^{a,b,c,}$\thanks{Supported by the National Natural Science Foundation of China (No. 11601431), the Natural Science Foundation of Shaanxi Province (No. 2020JM-099) and the Natural Science Foundation of Qinghai Province  (No. 2020-ZJ-920).}~$^,$\thanks{ Corresponding author. Email addresses: wang-wei@mail.nwpu.edu.cn, xiaogliu@nwpu.edu.cn, wj66@mail.nwpu.edu.cn}~,~Jing Wang$^{a,b}$
\\[2mm]
{\small $^a$School of Mathematics and Statistics,}\\[-0.8ex]
{\small Northwestern Polytechnical University, Xi'an, Shaanxi 710072, P.R.~China}\\
{\small $^b$Xi'an-Budapest Joint Research Center for Combinatorics,}\\[-0.8ex]
{\small Northwestern Polytechnical University, Xi'an, Shaanxi 710129, P.R. China}\\
{\small $^c$School of Mathematics and Statistics,}\\[-0.8ex]
{\small Qinghai Nationalities University, Xining, Qinghai 810007, P.R. China}
}

\date{}

\openup 0.5\jot
\maketitle

\begin{abstract}
In this paper, we give some sufficient conditions for graphs with an edge perturbation between twin vertices to have Laplacian perfect pair state transfer as well as Laplacian pretty good pair state transfer. By those sufficient conditions, we also construct many new graphs having Laplacian perfect pair state transfer as well as Laplacian pretty good pair state transfer.

\smallskip


\end{abstract}

\section{Introduction}

Let $G$ be a graph with the adjacency matrix $A_G$. Let $L_G=D_G-A_G$ denote the Laplacian matrix of $G$, where $D_G$ is the diagonal matrix with diagonal entries the degrees of vertices of $G$. The \emph{transition matrix} of $G$ relative to $L_{G}$ is defined by
$$
U_{L_{G}}(t) = \exp(-\mathrm{i}tL_{G})=\sum_{k=0}^{\infty}\frac{(-\mathrm{i})^{k}L_{G}^{k}t^{k}}{k!}, ~ t \in \mathbb{R},~\mathrm{i}=\sqrt{-1}.
$$
If $a$ and $b$ are two distinct vertices in $G$, we usually use $\{a,b\}$ to denote such a pair of vertices. Let $\mathbf{e}_{a}$ denote the unit vector with only the $a$-th entry equal to $1$ and all other entries equal to $0$. Then $\mathbf{e}_a-\mathbf{e}_b$ is called the \emph{pair state} of $\{a,b\}$. In particular, if $a$ and $b$ are connected by an edge in $G$, then such a pair state is called an \emph{edge state}. We say that a graph $G$ has \emph{Laplacian perfect pair state transfer} (Pair-LPST for short) \cite{QCh18, ChG19} from $\{a,b\}$ to $\{c,d\}$ at time $\tau$ if and only if
\begin{align*}
U_{L_{G}}(\tau)(\mathbf{e}_{a}-\mathbf{e}_{b})=\gamma(\mathbf{e}_{c}-\mathbf{e}_{d}),
\end{align*}
where $\gamma$ is a complex scalar and $\left|\gamma\right|=1$, which is equivalent to
\begin{align*}
\left|\frac{1}{2}(\mathbf e_{a}-\mathbf e_{b})^\top e^{-\mathrm{i}\tau L_{G}}(\mathbf e_{c}-\mathbf e_{d})\right|^{2}=1,
\end{align*}
where $\ast^\top$ denotes the transpose of $\ast$.
In this case,  we also say that $\mathbf e_{a}-\mathbf e_{b}$ (respectively, $\mathbf e_{c}-\mathbf e_{d}$)  has  Pair-LPST; otherwise, $\mathbf e_{a}-\mathbf e_{b}$ (respectively, $\mathbf e_{c}-\mathbf e_{d}$)  has  no Pair-LPST. In particular, if $\{a,b\}$ and $\{c,d\}$ are same, we say that $\{a,b\}$ (or $\mathbf e_{a}-\mathbf e_{b}$) is \emph{periodic} at time $\tau$ if and only if
\begin{align*}
\left|\frac{1}{2}(\mathbf e_{a}-\mathbf e_{b})^\top e^{-\mathrm{i}\tau L_{G}}(\mathbf e_{a}-\mathbf e_{b})\right|^{2}=1.
\end{align*}
Note that Pair-LPST can not always exist in a graph. Thus, a relaxation of Pair-LPST, the \emph{Laplacian pretty good pair state transfer} (short for Pair-LPGST), should be proposed. We say a graph $G$ has Pair-LPGST between $\{a,b\}$ and $\{c,d\}$ if for any $\epsilon>0$, there exists a time $\tau$ such that
\begin{align*}
\left|\frac{1}{2}(\mathbf e_{a}-\mathbf e_{b})^\top e^{-\mathrm{i}\tau L_{G}}(\mathbf e_{c}-\mathbf e_{d})\right|^{2}\geq 1-\epsilon.
\end{align*}
If both $\mathbf e_{a}-\mathbf e_{b}$ and $\mathbf e_{c}-\mathbf e_{d}$ are edge states in above definitions, then Pair-LPST and Pair-LPGST are also called Edge-LPST and Edge-LPGST, respectively.


Pair state transfers can be regarded as generalizations of state transfers, which are very important in quantum computing and quantum information processing \cite{Kay2010, ChristandlVZ2017, KemptonLY2017}. Up until now, many graphs have been shown to admit state transfers, including trees \cite{Bose03, CoutinhoL2015, Fan, GodsilKSS12}, Cayley graphs \cite{Basic11, CaoCL20, CaoF21, CaoWF20, CC,  HPal3, HPal5, HPal2021, Tan19, Tm19}, distance regular graphs \cite{Coutinho15} and some graph operations such as NEPS \cite{LiLZZ21, HPal1, HPal4, SZ}, coronas \cite{Ack, AckBCMT16}, edge complemented coronas \cite{WangL2021} and total graphs \cite{LiuW2021}. For more information, we refer the reader to \cite{Coutinho14, CGodsil, Godsil12, HZ, Zhou14}. Numerical results in \cite{ChG19} have shown that there should exist more graphs having pair state transfers than state transfers for a certain number of vertices. However, only few results on pair state transfers are given now. In 2018, Chen first proposed the concept of Edge-LPST in his thesis \cite{QCh18}, he presented many results on Edge-LPST, including the existence of Edge-LPST in paths and cycles, and he also gave two methods to construct new graphs with Edge-LPST. Those results can also be find in the published Paper \cite{ChG19}, where Pair-LPST was used instead of Edge-LPST. In 2021, Luo, Cao, et al. \cite{LCXC21} gave necessary and sufficient conditions for the existence of Edge-LPST in Cayley graphs of dihedral groups. With these conditions, the authors also proposed several concrete constructions of Cayley graphs admitting Edge-LPST. In the same year,  Cao \cite{Cao2021}  gave a characterization of cubelike graphs having Edge-LPST. The author also gave some concrete constructions of graphs admitting Edge-LPST. These seem to be the newest results on Pair-LPST.

Let $a$ and $b$ be two distinct vertices in a graph $G$. The edge perturbed graph obtained from $G$ by increasing the weight of the edge between $a$ and $b$ by $\alpha$ is denoted by $G+\alpha\{a,b\}$, $\alpha\in \mathbb{R}$. In particular, if $a$ and $b$ are adjacent vertices, then $G-\{a,b\}$ denotes the graph obtained from $G$ by deleting the edge between $a$ and $b$. If $a$ and $b$ are non-adjacent vertices, then $G+\{a,b\}$ denotes the graph obtained from $G$ by adding the edge between $a$ and $b$. Define
$$
M=(\mathbf e_{a}-\mathbf e_{b}) (\mathbf e_{a}-\mathbf e_{b})^\top.
$$
Then the Laplacian matrix of $G+\alpha\{a,b\}$ is denoted by
$$
L_{G+\alpha\{a,b\}}=L_G+\alpha M.
$$
Let $N(a)$ denote the set of all neighbours of a vertex $a$ in $G$. If $a$ and $b$ are two distinct vertices and
$$
N(a)\setminus\{b\}=N(b)\setminus\{a\},
$$
then $a$ and $b$ are called \emph{twin vertices}.

In this paper, we investigate the existence of Pair-LPST and Pair-LPGST in graphs with an edge perturbation between twin vertices. We give some sufficient conditions for edge perturbed graphs to have Pair-LPST as well as Pair-LPGST. By those sufficient conditions, we construct many new graphs having Pair-LPST as well as Pair-LPGST.





\section{Pair-LPST in edge perturbed graphs}
\label{Sec:LPST-EPG-1}

In this section, we investigate the existence of Pair-LPST in graphs with an edge perturbation between twin vertices.

\begin{lemma}\emph{(See \cite[Proposition~1]{Pal2022})}\label{TwinTrans-1}
 Let $G$ be a graph on $n$ vertices. Suppose that $a$ and $b$ are a pair of twin vertices in $G$.  Let $M=\left(\mathbf{e}_{a}-\mathbf{e}_{b}\right)\left(\mathbf{e}_{a}-\mathbf{e}_{b}\right)^{T}$ and $\alpha \in \mathbb{R}$. Then the transition matrix of the edge perturbed graph $G+\alpha\{a,b\}$ relative to $L_{G+\alpha\{a,b\}}$ is given by
$$
U_{L_{G+\alpha\{a,b\}}}(t)=U_{L_G}(t)\left(I_n+\frac{1}{2}(\exp (-2 \mathrm{i} \alpha t)-1) M\right),
$$
where $U_{L_G}(t)$ is the transition matrix of the unperturbed graph $G$ relative to $L_{G}$, and $I_n$ is the identity matrix of size $n$.
\end{lemma}

\begin{theorem}\label{Pair-LPST-Twin-1}
Let $G$, $a$, $b$, $M$ and $\alpha$ be as in Lemma \ref{TwinTrans-1}. Suppose that $\mathbf{e}_{p}-\mathbf{e}_{q}$ has Pair-LPST at time $\tau$ in the unperturbed graph $G$.
\begin{itemize}
  \item[\rm (a)] The pair state $\mathbf{e}_{a}-\mathbf{e}_{b}$ has no Pair-LPST at time $\tau$ in the edge perturbed graph $G+\alpha\{a,b\}$.
  \item[\rm (b)] If $p=a$ and $q\not=a, b$, then $\mathbf{e}_{p}-\mathbf{e}_{q}$ has Pair-LPST at time $\tau$ in the edge perturbed graph $G+\alpha\{a,b\}$ provided that $\alpha\tau\in \pi\mathbb{Z}$.
  \item[\rm (c)] If $p\not=a, b$ and $q\not=a, b$, then $\mathbf{e}_{p}-\mathbf{e}_{q}$ has Pair-LPST at time $\tau$ in the edge perturbed graph $G+\alpha\{a,b\}$.
\end{itemize}
\end{theorem}

\begin{proof}
(a) Note that
$$
M(\mathbf{e}_{a}-\mathbf{e}_{b}) =\left(\mathbf{e}_{a}-\mathbf{e}_{b}\right) \left(\mathbf{e}_{a}-\mathbf{e}_{b}\right)^{T} \left(\mathbf{e}_{a}-\mathbf{e}_{b}\right) =2\left(\mathbf{e}_{a}-\mathbf{e}_{b}\right).
$$
Then
\begin{align*}
  U_{L_{G+\alpha\{a,b\}}}(t)\left(\mathbf{e}_{a}-\mathbf{e}_{b}\right) &=U_{L_G}(t)\left(I_n+\frac{1}{2}(\exp (-2 \mathrm{i} \alpha t)-1) M\right)\left(\mathbf{e}_{a}-\mathbf{e}_{b}\right)\\
  &=\exp (-2 \mathrm{i} \alpha t)U_{L_G}(t)\left(\mathbf{e}_{a}-\mathbf{e}_{b}\right).
\end{align*}
Note \cite[Theorem 3.12]{ChG19} that if $a$ and $b$ are twin vertices, then
$\mathbf{e}_{a}-\mathbf{e}_{b}$ can never have Pair-LPST in $G$. Thus, by the definition of Pair-LPST, (a) can be verified immediately.

(b) In this case,
$$
M(\mathbf{e}_{a}-\mathbf{e}_{q}) =\left(\mathbf{e}_{a}-\mathbf{e}_{b}\right) \left(\mathbf{e}_{a}-\mathbf{e}_{b}\right)^{T} \left(\mathbf{e}_{a}-\mathbf{e}_{q}\right) =\left(\mathbf{e}_{a}-\mathbf{e}_{b}\right).
$$
Thus, if $\alpha\tau\in \pi\mathbb{Z}$, then
\begin{align*}
  U_{L_{G+\alpha\{a,b\}}}(\tau)\left(\mathbf{e}_{a}-\mathbf{e}_{q}\right) &=U_{L_G}(\tau)\left(I_n+\frac{1}{2}(\exp (-2 \mathrm{i} \alpha\tau)-1) M\right)\left(\mathbf{e}_{a}-\mathbf{e}_{q}\right)\\
  &=U_{L_G}(\tau)\left(\mathbf{e}_{a}-\mathbf{e}_{q}\right).
\end{align*}
By the definition of Pair-LPST, $\mathbf{e}_{a}-\mathbf{e}_{q}$ has Pair-LPST at time $\tau$ in the edge perturbed graph $G+\alpha\{a,b\}$.

(c) In this case,
$$
M(\mathbf{e}_{p}-\mathbf{e}_{q}) =\left(\mathbf{e}_{a}-\mathbf{e}_{b}\right) \left(\mathbf{e}_{a}-\mathbf{e}_{b}\right)^{T} \left(\mathbf{e}_{p}-\mathbf{e}_{q}\right) =0.
$$
Then
\begin{align*}
  U_{L_{G+\alpha\{a,b\}}}(t)\left(\mathbf{e}_{p}-\mathbf{e}_{q}\right) &= U_{L_G}(t)\left(\mathbf{e}_{p}-\mathbf{e}_{q}\right).
\end{align*}
Thus, by the definition of Pair-LPST again, $\mathbf{e}_{p}-\mathbf{e}_{q}$ has Pair-LPST at time $\tau$ in the edge perturbed graph $G+\alpha\{a,b\}$.
\qed
\end{proof}



\begin{example}
{\em Let $K_{2,4n}$ denote a complete bipartite graph consisting of an independent set of $2$ vertices completely joined to an independent set of $4n$ vertices, shown in Figure \ref{ComBipGra-Fig-1}. It was shown \cite[Theorem 5.2.2]{QCh18} that $K_{2,4n}$ has Pair-LPST from  $\mathbf{e}_a-\mathbf{e}_i$ to $\mathbf{e}_b- \mathbf{e}_i$  for $i \in \{1, 2, \ldots, 4n\}$ at time $\tau=\pi/2$.

\begin{figure}[H]
\begin{center}
\begin{tikzpicture}[x=0.75cm, y=0.55cm]
\tikzstyle{vertex}=[circle,inner sep=1.6pt, minimum size=0.1pt]

\vertex (a)[fill] at (0,0)[label=below:$1$]{};
\vertex (b)[fill] at (2,0)[label=below:$2$]{};
\vertex (c)[fill] at (5,0)[label=below:$4n-1$]{};
\vertex (d)[fill] at (7,0)[label=below:$4n$]{};
\vertex (g1) at (3.5,0.5)[label=below:$\ldots$]{};

\vertex (e)[fill] at (2.4,4)[label=above:$a$]{};
\vertex (f)[fill] at (4.6,4)[label=above:$b$]{};

\draw[line width=.3mm,line cap=round](a)--(e);
\draw[line width=.3mm,line cap=round](a)--(f);
\draw[line width=.3mm,line cap=round](b)--(e);
\draw[line width=.3mm,line cap=round](b)--(f);
\draw[line width=.3mm,line cap=round](c)--(e);
\draw[line width=.3mm,line cap=round](c)--(f);
\draw[line width=.3mm,line cap=round](d)--(e);
\draw[line width=.3mm,line cap=round](d)--(f);
\end{tikzpicture}
\end{center}
\vspace{-0.3cm}
\caption{The complete bipartite graph $K_{2,4n}$}
\label{ComBipGra-Fig-1}
\end{figure}
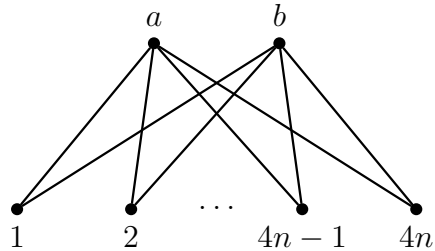

If we add an edge between $a$ and $b$ with weight $\alpha=2k$ with $k\in \mathbb{Z}$, shown in Figure \ref{ComBipGra-Fig-2}, then $\alpha\tau=k\pi\in\pi\mathbb{Z}$. By Theorem \ref{Pair-LPST-Twin-1} (b), the edge perturbed graph $K_{2,4n}+\alpha\{a,b\}$  has Pair-LPST from $\mathbf{e}_a-\mathbf{e}_i$ to $\mathbf{e}_b- \mathbf{e}_i$  for $i \in \{1, 2, \ldots, 4n\}$ at time $\tau=\pi/2$.

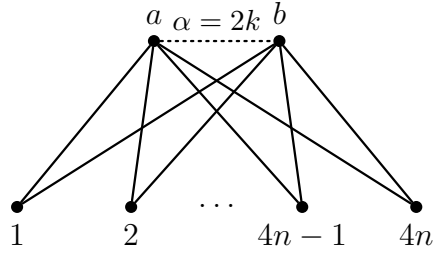
\begin{figure}[H]
\begin{center}
\begin{tikzpicture}[x=0.75cm, y=0.55cm]
\tikzstyle{vertex}=[circle,inner sep=1.6pt, minimum size=0.1pt]

\vertex (a)[fill] at (0,0)[label=below:$1$]{};
\vertex (b)[fill] at (2,0)[label=below:$2$]{};
\vertex (c)[fill] at (5,0)[label=below:$4n-1$]{};
\vertex (d)[fill] at (7,0)[label=below:$4n$]{};
\vertex (g1) at (3.5,0.5)[label=below:$\ldots$]{};

\vertex (e)[fill] at (2.4,4)[label=above:$a$]{};
\vertex (f)[fill] at (4.6,4)[label=above:$b$]{};

\draw[line width=.3mm,line cap=round](a)--(e);
\draw[line width=.3mm,line cap=round](a)--(f);
\draw[line width=.3mm,line cap=round](b)--(e);
\draw[line width=.3mm,line cap=round](b)--(f);
\draw[line width=.3mm,line cap=round](c)--(e);
\draw[line width=.3mm,line cap=round](c)--(f);
\draw[line width=.3mm,line cap=round](d)--(e);
\draw[line width=.3mm,line cap=round](d)--(f);

\draw[dotted, line width=.3mm,line cap=round](e)--(f);
\vertex (g2) at (3.5,3.85)[label=above:\small{$\alpha=2k$}]{};

\end{tikzpicture}
\end{center}
\vspace{-0.3cm}
\caption{The edge perturbed graph $K_{2,4n}+\alpha\{a,b\}$}
\label{ComBipGra-Fig-2}
\end{figure}

If we add an edge between $1$ and $2$, shown in Figure \ref{ComBipGra-Fig-3}, by Theorem \ref{Pair-LPST-Twin-1} (c), the edge perturbed graph $K_{2,4n}+ \{1,2\}$  has Pair-LPST from $\mathbf{e}_a-\mathbf{e}_i$ to $\mathbf{e}_b- \mathbf{e}_i$  for $i \in \{3, 4, \ldots, 4n\}$ at time $\tau=\pi/2$. Similarly, we can continue to add edges between other pair of vertices to obtain new graphs having Pair-LPST. Here, we omit the details, since one can easily achieve them by Theorem \ref{Pair-LPST-Twin-1} (c).

\begin{figure}[H]
\begin{center}
\begin{tikzpicture}[x=0.75cm, y=0.55cm]
\tikzstyle{vertex}=[circle,inner sep=1.6pt, minimum size=0.1pt]

\vertex (a)[fill] at (0,0)[label=below:$1$]{};
\vertex (b)[fill] at (2,0)[label=below:$2$]{};
\vertex (c)[fill] at (5,0)[label=below:$4n-1$]{};
\vertex (d)[fill] at (7,0)[label=below:$4n$]{};
\vertex (g1) at (3.5,0.5)[label=below:$\ldots$]{};

\vertex (e)[fill] at (2.4,4)[label=above:$a$]{};
\vertex (f)[fill] at (4.6,4)[label=above:$b$]{};

\draw[line width=.3mm,line cap=round](a)--(e);
\draw[line width=.3mm,line cap=round](a)--(f);
\draw[line width=.3mm,line cap=round](b)--(e);
\draw[line width=.3mm,line cap=round](b)--(f);
\draw[line width=.3mm,line cap=round](c)--(e);
\draw[line width=.3mm,line cap=round](c)--(f);
\draw[line width=.3mm,line cap=round](d)--(e);
\draw[line width=.3mm,line cap=round](d)--(f);

\draw[line width=.3mm,line cap=round](a)--(b);

\end{tikzpicture}
\end{center}
\vspace{-0.3cm}
\caption{The edge perturbed graph $K_{2,4n}+\{1,2\}$}
\label{ComBipGra-Fig-3}
\end{figure}
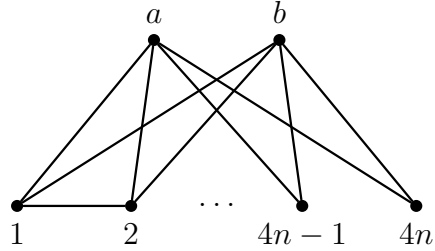

}
\end{example}



\begin{theorem}\label{Pair-LPST-Twin-2}
Let $G$, $a$, $b$, $M$ and $\alpha$ be as in Lemma \ref{TwinTrans-1}. Suppose that $\mathbf{e}_{p}-\mathbf{e}_{q}$ is periodic at time $\tau$ in the unperturbed graph $G$.
\begin{itemize}
  \item[\rm (a)] If $p=a$ and $q=b$, then $\mathbf{e}_{p}-\mathbf{e}_{q}$ is periodic at time $\tau$ in the edge perturbed graph $G+\alpha\{a,b\}$.
  \item[\rm (b)] If $p=a$, $q\not=a, b$ and $2\alpha\tau\in \pi(2\mathbb{Z}+1)$, then $\mathbf{e}_{p}-\mathbf{e}_{q}$ has Pair-LPST at time $\tau$ in the edge perturbed graph $G+\alpha\{a,b\}$. More precisely, $G+\alpha\{a,b\}$ has Pair-LPST at time $\tau$ from  $\mathbf{e}_{a}-\mathbf{e}_{q}$ to $\mathbf{e}_{b}-\mathbf{e}_{q}$ for $q\not=a, b$.
  \item[\rm (c)] If $p\not=a, b$ and $q\not=a, b$, then $\mathbf{e}_{p}-\mathbf{e}_{q}$ is periodic at time $\tau$ in the edge perturbed graph $G+\alpha\{a,b\}$.
\end{itemize}
\end{theorem}

\begin{proof}
From the proof of Theorem \ref{Pair-LPST-Twin-1} (a) and (c), we can immediately verify (a) and (c). For (b), note that
$$
M(\mathbf{e}_{b}-\mathbf{e}_{q}) =\left(\mathbf{e}_{a}-\mathbf{e}_{b}\right) \left(\mathbf{e}_{a}-\mathbf{e}_{b}\right)^{T} \left(\mathbf{e}_{b}-\mathbf{e}_{q}\right) =-\left(\mathbf{e}_{a}-\mathbf{e}_{b}\right).
$$
Thus, if $\alpha\tau\in \pi(2\mathbb{Z}+1)$, then
\begin{align*}
  U_{L_{G+\alpha\{a,b\}}}(\tau)\left(\mathbf{e}_{b}-\mathbf{e}_{q}\right) &=U_{L_G}(\tau)\left(I_n- M\right)\left(\mathbf{e}_{b}-\mathbf{e}_{q}\right)\\
  &=U_{L_G}(\tau)\left(\mathbf{e}_{a}-\mathbf{e}_{q}\right).
\end{align*}
By the definition of Pair-LPST, $G+\alpha\{a,b\}$ has Pair-LPST at time $\tau$ from  $\mathbf{e}_{a}-\mathbf{e}_{q}$ to $\mathbf{e}_{b}-\mathbf{e}_{q}$ for $q\not=a, b$.
\qed
\end{proof}

The following result was given in \cite[Corollary~5.4]{ChG19}, where the fact \cite[Theorem~5.2]{ChG19} that a pair state in a graph has Pair-LPST if and only if it has Pair-LPST in its complement was used to prove this result. Here, we use Theorem \ref{Pair-LPST-Twin-2} (b) to reconfirm this result.

\begin{cor}\label{CompleteDeletion-1}
Let $K_{n}$ be a complete graph on $n$ vertices with the vertex set $V\left(K_{n}\right)=\left\{1, 2, \ldots, n\right\}$. Let $a$ and $b$ be two distinct vertices of $K_{n}$. Then $K_{n}-\{a,b\}$ has Pair-LPST from $\mathbf{e}_{a}-\mathbf{e}_{q}$ to $\mathbf{e}_{b}-\mathbf{e}_{q}$ for all $q \in\left\{1, 2, \ldots, n\right\}\setminus\{a,b\}$.
\end{cor}

\begin{proof}
Note \cite[Corollary~5.1.5]{QCh18} that every edge in a complete graph is periodic  for any time $\tau$. Note also that $a$ and $b$ are twin vertices. Choose $\tau=\frac{\pi}{2}$ and $\alpha=-1$, then $2\alpha\tau=-\pi\in \pi(2\mathbb{Z}+1)$. By Theorem \ref{Pair-LPST-Twin-2} (b), $K_{n}-\{a,b\}$ has Pair-LPST from $\mathbf{e}_{a}-\mathbf{e}_{q}$ to $\mathbf{e}_{b}-\mathbf{e}_{q}$ at time $\tau=\frac{\pi}{2}$ for all $q \in\left\{1, 2, \ldots, n\right\}\setminus\{a,b\}$.
\qed
\end{proof}


More generally, by Theorem \ref{Pair-LPST-Twin-2}, we obtain the following result.

\begin{cor}
Let $K_{n}$ be a complete graph on $n$ vertices with the vertex set $V\left(K_{n}\right)=\left\{1, 2, \ldots, n\right\}$. Let $S$ be a matching of $K_n$, and let $K_n-S$ denote the graph obtained from $K_n$ by deleting all edges in $S$. Suppose that $a$ and $b$ are two end vertices of an edge in $S$. Then $K_n-S$ has Pair-LPST from $\mathbf{e}_{a}-\mathbf{e}_{q}$ to $\mathbf{e}_{b}-\mathbf{e}_{q}$ for all $q \in\left\{1, 2, \ldots, n\right\}\setminus\{a,b\}$.
\end{cor}

\begin{proof}
With a slight abuse of notation, let $a\sim b$ denote the edge with two end vertices $a$ and $b$. Let $c\sim d$ be an edge in $S$, different from $a\sim b$. Then $c$ and $d$ are twin vertices. By Theorem \ref{Pair-LPST-Twin-2} (c), for any edge $x\sim y$ in $S\setminus\{c\sim d\}$, $\mathbf{e}_{x}-\mathbf{e}_{y}$ is periodic in  $K_n-\{c,d\}$ at any time $\tau$. Let $f\sim g$ be another edge in $S$, different from $a\sim b$ and $c\sim d$. Then $f$ and $g$ are twin vertices. By Theorem \ref{Pair-LPST-Twin-2} (c), for any edge $x\sim y$ in $S\setminus\{c\sim d, f\sim g\}$, $\mathbf{e}_{x}-\mathbf{e}_{y}$ is periodic in  $K_n-\{c,d\}-\{f,g\}$ at any time $\tau$. Keeping on this operation, until we obtain the graph $K_n-S\setminus\{a\sim b\}$. Here, $\mathbf{e}_{a}-\mathbf{e}_{b}$ is periodic in $K_n-S\setminus\{a\sim b\}$ at any time $\tau$. And, $a$ and $b$ are twin vertices in $K_n-S\setminus\{a\sim b\}$. Similar to the proof of Corollary \ref{CompleteDeletion-1}, choose $\tau=\frac{\pi}{2}$ and $\alpha=-1$, then $2\alpha\tau=-\pi\in \pi(2\mathbb{Z}+1)$. By Theorem \ref{Pair-LPST-Twin-2} (b), $K_{n}-S$ has Pair-LPST from $\mathbf{e}_{a}-\mathbf{e}_{q}$ to $\mathbf{e}_{b}-\mathbf{e}_{q}$ at time $\tau=\frac{\pi}{2}$ for all $q \in\left\{1, 2, \ldots, n\right\}\setminus\{a,b\}$.
\qed
\end{proof}


\section{Pair-LPGST in edge perturbed graphs}
\label{Sec:LPST-EPG-2}

In this section, we give the existence of Pair-LPGST in graphs with an edge perturbation between twin vertices. By refining the proof of Theorem \ref{Pair-LPST-Twin-1}, one can easily verify the following result.

\begin{theorem}\label{Pair-LPGST-Twin-1}
Let $G$, $a$, $b$, $M$ and $\alpha$ be as in Lemma \ref{TwinTrans-1}. Suppose that $\mathbf{e}_{p}-\mathbf{e}_{q}$ has Pair-LPGST at time $\tau$ in the unperturbed graph $G$.
\begin{itemize}
  \item[\rm (a)] If $p=a$ and $q=b$, then $\mathbf{e}_{p}-\mathbf{e}_{q}$ has Pair-LPGST at time $\tau$ in the edge perturbed graph $G+\alpha\{a,b\}$.
  \item[\rm (b)] If $p=a$ and $q\not=a, b$, then $\mathbf{e}_{p}-\mathbf{e}_{q}$ has Pair-LPGST at time $\tau$ in the edge perturbed graph $G+\alpha\{a,b\}$ provided that $\alpha\tau\in \pi\mathbb{Z}$.
  \item[\rm (c)] If $p\not=a, b$ and $q\not=a, b$, then $\mathbf{e}_{p}-\mathbf{e}_{q}$ has Pair-LPGST at time $\tau$ in the edge perturbed graph $G+\alpha\{a,b\}$.
\end{itemize}
\end{theorem}

\begin{example}
{\em Let $\Gamma$ be a group with the identity element $1_e$. Let $S$ be a subset of $\Gamma$ satisfying that $1_e\notin S$ and $S=S^{-1}=\left\{s^{-1}~|~s\in S\right\}$. The \emph{Cayley graph} $\Cay(\Gamma,S)$ is defined to be the graph with the vertex set $\Gamma$ and with vertices $a$ and $b$ adjacent if and only if $ ab^{-1}\in S$. By refining the proof of \cite[Lemma~3.1]{LCXC21}, we obtain that if $a$ and $b$ are two distinct vertices of $\Cay(\Gamma,S)$ and $\Cay(\Gamma,S)$ has \emph{Laplacian pretty good state transfer} (LPGST) from $a$ to $b$ at time $\tau$, then $\Cay(\Gamma,S)$ has Pair-LPGST from $\mathbf{e}_a-\mathbf{e}_{a\gamma}$ to $\mathbf{e}_b-\mathbf{e}_{b\gamma}$ at time $\tau$, for any $\gamma\in \Gamma\setminus\{a^{-1}b, b^{-1}a\}$.

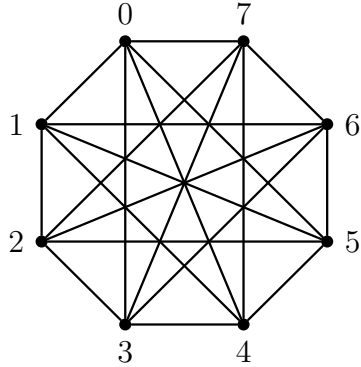
\begin{figure}[H]
\begin{center}
\begin{tikzpicture}[x=1.1cm, y=1.1cm]
\tikzstyle{vertex}=[circle,inner sep=1.6pt, minimum size=0.1pt]

\vertex (a)[fill] at (0,0)[label=above:$0$]{};
\vertex (b)[fill] at (-1,-1)[label=left:$1$]{};
\vertex (c)[fill] at (-1,-2.414)[label=left:$2$]{};
\vertex (d)[fill] at (0,-3.414)[label=below:$3$]{};
\vertex (e)[fill] at (1.414,-3.414)[label=below:$4$]{};
\vertex (f)[fill] at (2.414,-2.414)[label=right:$5$]{};
\vertex (g)[fill] at (2.414,-1)[label=right:$6$]{};
\vertex (h)[fill] at (1.414,0)[label=above:$7$]{};

\draw[line width=.3mm,line cap=round](g)--(h);
\draw[line width=.3mm,line cap=round](f)--(g);
\draw[line width=.3mm,line cap=round](e)--(f);
\draw[line width=.3mm,line cap=round](e)--(h);
\draw[line width=.3mm,line cap=round](d)--(e);
\draw[line width=.3mm,line cap=round](d)--(g);
\draw[line width=.3mm,line cap=round](d)--(h);
\draw[line width=.3mm,line cap=round](c)--(d);
\draw[line width=.3mm,line cap=round](c)--(f);
\draw[line width=.3mm,line cap=round](c)--(h);
\draw[line width=.3mm,line cap=round](c)--(g);
\draw[line width=.3mm,line cap=round](b)--(c);
\draw[line width=.3mm,line cap=round](b)--(e);
\draw[line width=.3mm,line cap=round](b)--(f);
\draw[line width=.3mm,line cap=round](b)--(g);
\draw[line width=.3mm,line cap=round](a)--(b);
\draw[line width=.3mm,line cap=round](a)--(d);
\draw[line width=.3mm,line cap=round](a)--(e);
\draw[line width=.3mm,line cap=round](a)--(f);
\draw[line width=.3mm,line cap=round](a)--(h);
\end{tikzpicture}
\end{center}
\vspace{-0.25cm}
\caption{The Cayley graph $\Cay(\mathbb{Z}_8,S)$ with $S=\{1,3,4,5,7\}$}
\label{CayleyExam-Fig-1}
\end{figure}

Let $\Gamma=\mathbb{Z}_8$ and $S=\{1,3,4,5,7\}$. By \cite[Corollary 2.9]{HPal2021}, the Cayley graph $\Cay(\mathbb{Z}_8,S)$, shown in Figure \ref{CayleyExam-Fig-1}, has LPGST from $a$ to $a+4$ at time $\tau=\frac{\pi}{2}$, for any $a\in \mathbb{Z}_8$. Choose $a=0$ and $\gamma=1$. Then $\Cay(\mathbb{Z}_8,S)$ has Pair-LPGST from $\mathbf{e}_0-\mathbf{e}_{1}$ to $\mathbf{e}_4-\mathbf{e}_{5}$ at time $\tau=\frac{\pi}{2}$. Note that $2$ and $6$ are twin vertices. Thus, by Theorem \ref{Pair-LPGST-Twin-1} (c),  the edge perturbed graph  $\Cay(\mathbb{Z}_8,S)-\{2, 6\}$, shown in Figure \ref{CayleyExam-Fig-2}, has Pair-LPGST from $\mathbf{e}_0-\mathbf{e}_{1}$ to $\mathbf{e}_4-\mathbf{e}_{5}$ at time $\tau=\frac{\pi}{2}$.

\begin{figure}[H]
\begin{center}
\begin{tikzpicture}[x=1.1cm, y=1.1cm]
\tikzstyle{vertex}=[circle,inner sep=1.6pt, minimum size=0.1pt]

\vertex (a)[fill] at (0,0)[label=above:$0$]{};
\vertex (b)[fill] at (-1,-1)[label=left:$1$]{};
\vertex (c)[fill] at (-1,-2.414)[label=left:$2$]{};
\vertex (d)[fill] at (0,-3.414)[label=below:$3$]{};
\vertex (e)[fill] at (1.414,-3.414)[label=below:$4$]{};
\vertex (f)[fill] at (2.414,-2.414)[label=right:$5$]{};
\vertex (g)[fill] at (2.414,-1)[label=right:$6$]{};
\vertex (h)[fill] at (1.414,0)[label=above:$7$]{};

\draw[line width=.3mm,line cap=round](g)--(h);
\draw[line width=.3mm,line cap=round](f)--(g);
\draw[line width=.3mm,line cap=round](e)--(f);
\draw[line width=.3mm,line cap=round](e)--(h);
\draw[line width=.3mm,line cap=round](d)--(e);
\draw[line width=.3mm,line cap=round](d)--(g);
\draw[line width=.3mm,line cap=round](d)--(h);
\draw[line width=.3mm,line cap=round](c)--(d);
\draw[line width=.3mm,line cap=round](c)--(f);
\draw[line width=.3mm,line cap=round](c)--(h);
%
\draw[line width=.3mm,line cap=round](b)--(c);
\draw[line width=.3mm,line cap=round](b)--(e);
\draw[line width=.3mm,line cap=round](b)--(f);
\draw[line width=.3mm,line cap=round](b)--(g);
\draw[line width=.3mm,line cap=round](a)--(b);
\draw[line width=.3mm,line cap=round](a)--(d);
\draw[line width=.3mm,line cap=round](a)--(e);
\draw[line width=.3mm,line cap=round](a)--(f);
\draw[line width=.3mm,line cap=round](a)--(h);
\end{tikzpicture}
\end{center}
\vspace{-0.25cm}
\caption{The edge perturbed graph $\Cay(\mathbb{Z}_8,S)-\{2, 6\}$  with $S=\{1,3,4,5,7\}$}
\label{CayleyExam-Fig-2}
\end{figure}
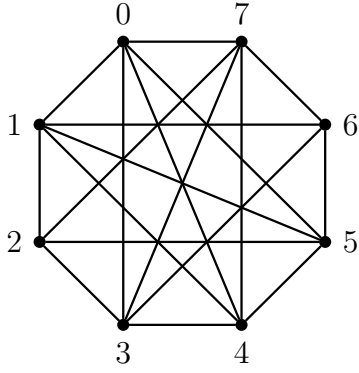

Similarly, by Theorem \ref{Pair-LPGST-Twin-1} (c), we can continue to delete the edge between vertices $3$ and $7$ to obtain a new graph having Pair-LPGST from $\mathbf{e}_0-\mathbf{e}_{1}$ to $\mathbf{e}_4-\mathbf{e}_{5}$ at time $\tau=\frac{\pi}{2}$,  shown in Figure \ref{CayleyExam-Fig-3}.

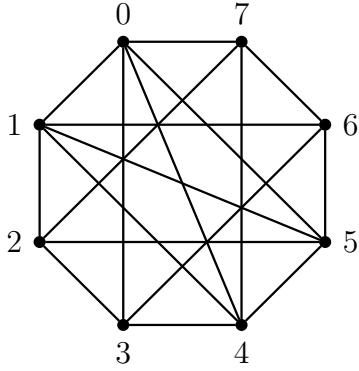
\begin{figure}[H]
\begin{center}
\begin{tikzpicture}[x=1.1cm, y=1.1cm]
\tikzstyle{vertex}=[circle,inner sep=1.6pt, minimum size=0.1pt]

\vertex (a)[fill] at (0,0)[label=above:$0$]{};
\vertex (b)[fill] at (-1,-1)[label=left:$1$]{};
\vertex (c)[fill] at (-1,-2.414)[label=left:$2$]{};
\vertex (d)[fill] at (0,-3.414)[label=below:$3$]{};
\vertex (e)[fill] at (1.414,-3.414)[label=below:$4$]{};
\vertex (f)[fill] at (2.414,-2.414)[label=right:$5$]{};
\vertex (g)[fill] at (2.414,-1)[label=right:$6$]{};
\vertex (h)[fill] at (1.414,0)[label=above:$7$]{};

\draw[line width=.3mm,line cap=round](g)--(h);
\draw[line width=.3mm,line cap=round](f)--(g);
\draw[line width=.3mm,line cap=round](e)--(f);
\draw[line width=.3mm,line cap=round](e)--(h);
\draw[line width=.3mm,line cap=round](d)--(e);
\draw[line width=.3mm,line cap=round](d)--(g);
%
\draw[line width=.3mm,line cap=round](c)--(d);
\draw[line width=.3mm,line cap=round](c)--(f);
\draw[line width=.3mm,line cap=round](c)--(h);
%
\draw[line width=.3mm,line cap=round](b)--(c);
\draw[line width=.3mm,line cap=round](b)--(e);
\draw[line width=.3mm,line cap=round](b)--(f);
\draw[line width=.3mm,line cap=round](b)--(g);
\draw[line width=.3mm,line cap=round](a)--(b);
\draw[line width=.3mm,line cap=round](a)--(d);
\draw[line width=.3mm,line cap=round](a)--(e);
\draw[line width=.3mm,line cap=round](a)--(f);
\draw[line width=.3mm,line cap=round](a)--(h);
\end{tikzpicture}
\end{center}
\vspace{-0.25cm}
\caption{The edge perturbed graph $\Cay(\mathbb{Z}_8,S)-\{2, 6\}-\{3,7\}$  with $S=\{1,3,4,5,7\}$}
\label{CayleyExam-Fig-3}
\end{figure}

}
\end{example}

\section{Conclusions}

As stated in the Introduction, up until now, only few graphs have been proved to have Pair-LPST. In this paper, we give some sufficient conditions for edge perturbed graphs to have Pair-LPST as well as Pair-LPGST. Those sufficient conditions can enable us to construct infinitely many new graphs having Pair-LPST as well as Pair-LPGST. In this paper, we only construct a few graphs having Pair-LPST as well as Pair-LPGST. Thus, other examples constructed by those sufficient conditions are also greatly welcome.


\bigskip

\noindent \textbf{Acknowledgements}~~The authors greatly appreciate Professor Hiranmoy Pal who inspires the authors to complete this paper.




\end{document}